\documentclass[sn-mathphys]{sn-jnl}

\jyear{2021}%

\theoremstyle{thmstyleone}%
\newtheorem{theorem}{Theorem}
\theoremstyle{thmstyletwo}%
\newtheorem{remark}{Remark}%

\theoremstyle{thmstylethree}%

\newcommand{\bes}{\begin{equation*}}
\newcommand{\ees}{\end{equation*}}
\newcommand{\beq}{\begin{equation}}
\newcommand{\eeq}{\end{equation}}

\usepackage{multirow}
\usepackage{xcolor}

\usepackage[makeroom]{cancel}

\usepackage{xcolor}

\usepackage[makeroom]{cancel}
\newcommand{\diff}[2]{\frac{\mathrm {#1}}{\mathrm{#1}#2}}
\newcommand{\barr}{\begin{array}}
\newcommand{\earr}{\end{array}}

\usepackage[right]{lineno}
\usepackage{microtype}
\DisableLigatures[f]{encoding = *, family = * }

\usepackage{array}

\newcolumntype{+}{!{\vrule width 2pt}}

\newlength\savedwidth

\newcommand\thickhline{\noalign{\global\savedwidth\arrayrulewidth\global\arrayrulewidth 2pt}%
\hline
\noalign{\global\arrayrulewidth\savedwidth}}

\newcommand\scalemath[2]{\scalebox{#1}{\mbox{\ensuremath{\displaystyle #2}}}}

\begin{document}

\title[Understanding Biofilm-Phage Interactions in Cystic Fibrosis Patients Using Math]{Understanding Biofilm-Phage Interactions in Cystic Fibrosis Patients Using Mathematical Frameworks}


\author*[1]{\fnm{Blessing O.} \sur{Emerenini}}\email{boesma@rit.edu}

\author[2]{\fnm{Doris} \sur{Hartung}}\email{dthartung@alum.ship.edu}
\author[3]{\fnm{Ricardo N. G} \sur{Reyes Grimaldo}}\email{reyesgrr@oregonstate.edu}
\author[1]{\fnm{Claire}
\sur{Canner}}\email{cmc2511@rit.edu}

\author[4]{\fnm{Maya} \sur{Williams}}\email{willim54@tcnj.edu}

\author[1]{\fnm{Ephraim} \sur{Agyingi}}\email{eoasma@rit.edu}

\author[5]{\fnm{Robert} \sur{Osgood}}\email{rcoscl@rit.edu}


\affil[1]{\orgdiv{School of Mathematical Sciences}, \orgname{Rochester Institute of Technology}, \orgaddress{\street{One Lomb Memorial Drive}, \city{Rochester}, \postcode{14623}, \state{NY}, \country{USA}}}

\affil[2]{\orgdiv{Department of Mathematics, Shippensburg University}, \orgaddress{\street{1871 Old Main Dr}, \city{Shippensburg}, \postcode{17257}, \state{PA}, \country{USA}}}

\affil[3]{\orgdiv{Department of Integrative Biology, Oregon State University}, \orgaddress{\street{2701 SW Campus Way}, \city{Corvallis}, \postcode{97331}, \state{OR}, \country{USA}}}

\affil[4]{\orgdiv{Department of Mathematics}, \orgname{The College of New Jersey}, \orgaddress{\street{2000 Pennington Rd}, \city{Ewing Township}, \postcode{08618}, \state{NJ}, \country{USA}}}

\affil[5]{\orgdiv{College of Health Sciences and Technology}, \orgname{Rochester Institute of Technology}, \orgaddress{\street{One Lomb Memorial Drive}, \city{Rochester}, \postcode{14623}, \state{NY}, \country{USA}}}


\abstract{
    When planktonic bacteria adhere together to a surface, they begin to form biofilms, or 
    communities of bacteria. Biofilm formation in a host can be extremely problematic if left untreated, 
    especially since antibiotics can be ineffective in treating the bacteria. Certain lung diseases such 
    as cystic fibrosis can cause the formation of biofilms in the lungs and can be fatal. With 
    antibiotic-resistant bacteria, the use of phage therapy has been introduced as an alternative or an 
    additive to the use of antibiotics in order to combat biofilm growth. Phage therapy utilizes phages, 
    or viruses that attack bacteria, in order to penetrate and eradicate biofilms. In order to evaluate 
    the effectiveness of phage therapy against biofilm bacteria, we adapt an ordinary differential equation 
    model to describe the dynamics of phage-biofilm combat in the lungs. We then create our own 
    phage-biofilm model with ordinary differential equations and stochastic modeling. Then, simulations 
    of parameter alterations in both models are investigated to assess how they will affect the efficiency 
    of phage therapy against bacteria. By increasing the phage mortality rate, the biofilm growth can be 
    balanced and allow the biofilm to be more vulnerable to antibiotics. Thus, phage therapy is an effective 
    aid in biofilm treatment.
}

\keywords{Biofilm, Phages, mathematical modeling, stochastic modeling}



\maketitle

\section{Introduction}\label{sec1}

    Presence of pathogenic microorganisms in our environment entail enormous problems for humans and livestock. 
    The problem of pathogenic microrganisms is even grievous when they reside in host \cite{refr01}. 
    Bacteria is one of such pathogenic microorganisms and they prefer to live in communities called Biofilms. 

    Biofilms are aggregation of bacteria on immersed surfaces and interfaces, in which the cells are embedded 
    in a self-produced layer of extrecellular polymeric substances (EPS). The EPS gives them protection 
    against mechanical washout and antibiotics. The formation of biofilm is often considered a virulence 
    factor \cite{refr1}. Given the role that biofilms play in resistance, new treatments are promoted which 
    aim at penetrating the biofilm matrix and attacking the individual cells in the biofilm. Dissolved 
    growth-limiting substrates such as oxygen diffuse into the biofilm and undergo reaction with bacteria. 
    In many instances, in well-developed biofilms, such growth limiting substrates might only be able to 
    penetrate the biofilm over a relatively thin active outer layers and substantial inactive inner layers 
    may form. Several studies has shown that there are some immaterial substances and microbes which can 
    also penetrate the biofilm matrix, one of such microbes is the bacteriophages.

    Bacteriophage, also known informally as phage, is a virus that infects and replicates within bacteria 
    and archaea, it is among the most common disease entities in the biosphere. Bacteriophages are 
    exclusively used as therapeutic agents to treat infections caused by pathogenic bacteria. Application 
    of Phage therapy was dated more than a century ago with poor understanding of its potentials 
    \cite{refr2,refr3} and so was overshadowed in western medicine until the emergence of bacterial strains 
    which were resistant to antibiotics \cite{refr4}. The use of phage in treatments has a number of potential 
    advantages over the use of antibiotics \cite{refr5, refr6}. 

    This study focuses on the interaction between biofilm and bacteriophage in phage therapy. There are two 
    main approaches to studying phage-biofilm interactions: experimental approach, and the mathematical 
    modeling approach. Both methods are widely utilized and each have advantages. Using mathematical models 
    as a way to study disease dynamics is an extremely useful tool when it comes to the observation and 
    prevention of infections in humans \cite{refr07, refr08, refr8, refr9}. Through the use of mathematical 
    modeling, one is able to implement approaches in which infectious outbreaks can be predicted, assessed, and 
    controlled \cite{refr7, refr8, refr9}. In order to model these dynamics a form of model must be chosen. 
    These mathematical models can range from equation-based modeling, such as ordinary, partial, and stochastic 
    differential equations, to agent-based modeling \cite{refr10, refr11}. From different studies that utilized 
    mathematical modeling, findings have been made on in-host disease dynamics. Through the use of differential 
    equations in a study conducted by Beke et. al., the disease dynamics in a bacteria-phage interaction are 
    examined. The results of this study gathered that disease dynamics can be contingent on environmental 
    factors such as a change in pH or temperature \cite{refr12}. 
    In a separate study conducted by Bardina et. al., a different approach involving a stochastic model is used 
    in order to analyze these dynamics. Upon these findings, there existed equilibria such that pathogens could 
    be eliminated from the host or could persist depending on the levels of noise within the environment 
    \cite{refr13}. Similarly, through differential equations and Monte Carlo simulations, Sinha et. al. 
    evaluate which mathematical model is adequate for modeling the dynamics between phage and bacteria. In this 
    study, it was found that disease dynamics can differ if there are spatial restrictions introduced to the 
    model and the type of model used to describe such dynamics should reflect this restriction \cite{refr14}. 
    These mathematical models are not only used for examining in-host dynamics, but are also used to address 
    various phenomena. Some of these phenomena include spatial phenomena as previously mentioned, 
    evolutionary game theory, and dynamic optimization.

    Mathematical models for bacterial biofilm over the years have greatly helped in the understanding of 
    biofilm processes such as biofilm formation and growth; detachment and its inducers 
    \cite{refr15, refr015, refr16, refr17, refr18} . Many of the experimental and modeling studies of 
    biofilm-phage interactions and interplay has focused on biofilm formed on surfaces other than in-host, 
    mathematical modeling that focus on the biofilm-phage interactions and interplay in immunocompromised 
    patents is still in its infancy. There are several immunocompromised patents that suffer from biofilm 
    infection, we will be considering the case of biofilm formation in the lungs of Cystic Fibrosis patients.

    Over the last two decades, a large number of models have been produced in order to represent the 
    interactions between bacteria, biofilm, and bacteriophages. Some of these models involve ODEs, PDEs, 
    agent based modeling, or stochastic models, but all of which attempt to recreate results that could be 
    produced in an experiment in order to better and more quickly understand and predict these interactions. 
    In a recent review from Sinha, et. al., we find models utilizing PDE's, ODE's, stochastic differential 
    equations and Monte Carlo simulations. This review compares models subject to spatial constraints with 
    those without spatial constraints \cite{refr14}. 

    In this study, the objective is to develop a mathematical framework to understand the different factors 
    that contribute most during bacterial-phage interactions in biofilm setting and planktonic phage.
\section{Methods}
    \subsection{\emph{Basic model assumptions}}
        We develop a mathematical model that describes the dynamics of biofilm-phage interactions in matured 
        biofilms. There are two specific regions involved in this system, namely: the Biofilm region, this 
        is the region where the bacterial cells are accumulated and formed biofilm;  and the Planktonic region, 
        this is the flow region of the bronchiole comprising of air and fluid (see 
        \autoref{fig:schematicbiofilm}). We assume that (a) both bacteria and phages can be converted from 
        one region to the other (b) conversion of biofilm bacteria to planktonic bacteria is induced by 
        phages (c) phages conversion rate between biofilm and the planktonic phase is assumed to be constant, 
        and (d) the conversion of phages from one region to the other does not change their characteristics. 

        \begin{figure}[!ht]
            \centering
            \includegraphics[scale=0.25]{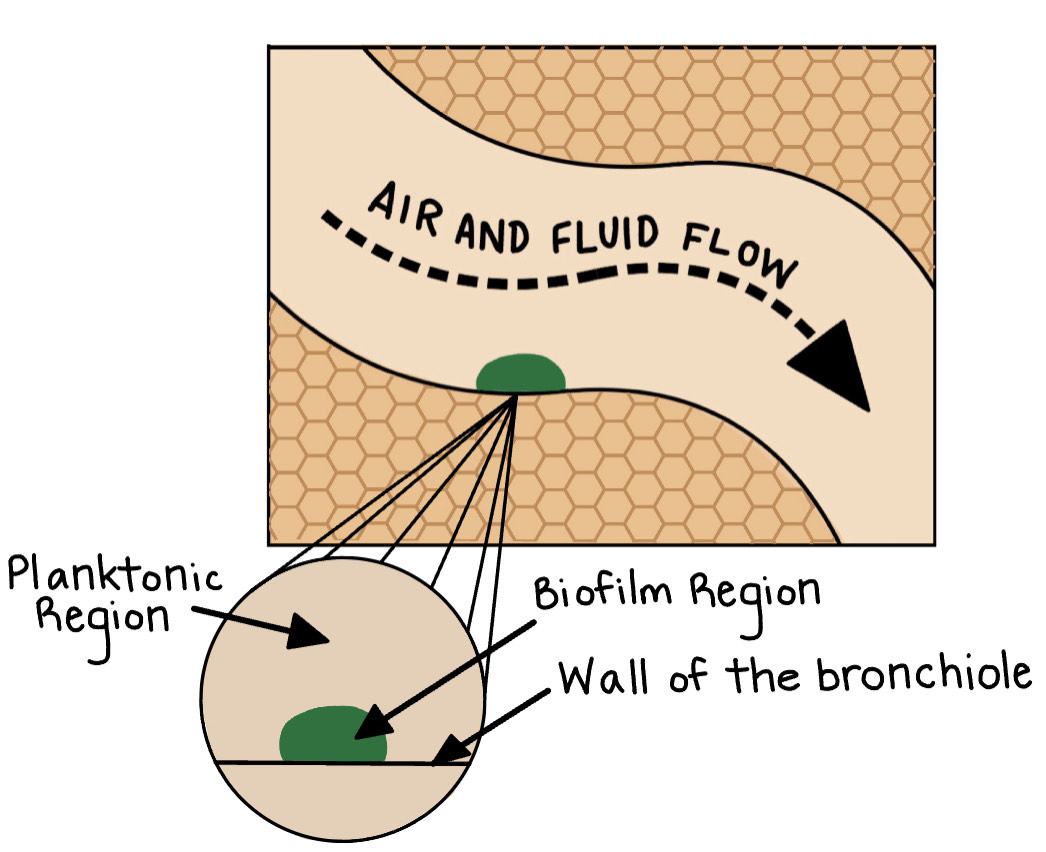}
            \caption{{\bf Schematic of Biofilm in Bronchiole} This is a schematic representation of the 
            biofilm that formed on the wall of the bronchiole, showing the biofilm region and the 
            planktonic region}
            \label{fig:schematicbiofilm}
        \end{figure}

    \subsection{\emph{Deterministic Model - ODE}}\label{mathmodel}
        The model describing the Biofilm-Phage interactions is formulated as a deterministic model of as 
        system of six ordinary differential equations. The dependent variables are $n$, $B$, $P$, $V_B$, 
        $V_P$ and $I$. The variable $n$ denotes the concentration of the bacteria growth limiting nutrient 
        substrate, $B$ denotes the bacteria cells in the biofilm region while $P$ denotes the  bacterial cells 
        in the planktonic region, $V_B$ and $V_P$ denotes the viral load of bacteriophages in the biofilm 
        and the planktonic regions respectively; $I$ is the concentration of all the infected bacteria cells 
        from the planktonic and biofilm regions. The model captures the detachment of bacteria cells from the 
        biofilm induced by bacteriophages, and reattachment of bacteria to the biofilm, this meshes well with 
        the life cycle of biofilms. The governing equations read

        \begin{eqnarray}
        	\diff{d}{t}n&=&f(n)-\left(\lambda_B B+\lambda_P P\right)\frac{n}{n+k}
        	\label{eq:nutrients} \\
        	\diff{d}{t}B&=&\lambda_BB\frac{n}{n+k}-\phi_1BV_B
        	-\gamma_1\frac{V_B}{\zeta_1+V_B}B+\gamma_2\frac{B}{\zeta_2+B}P-\mu_BB
        	\label{eq:cellbiofilm} \\
        	\diff{d}{t}P&=&\lambda_PP\frac{n}{n+k}-\phi_2PV_P
        	+\gamma_1\frac{V_B}{\zeta_1+V_B}B-\gamma_2\frac{B}{\zeta_2+B}P-\mu_PP
        	\label{eq:cellplate}  \\
        	\diff{d}{t}V_B&=&\beta\phi_1V_BB-c_1V_B-qV_B+pV_P\label{eq:phagebiofilm}  \\
        	\diff{d}{t}V_P&=&\beta\phi_2V_PP-c_2V_P+qV_B-pV_P\label{eq:phageplate} \\
        	\diff{d}{t}I&=&\phi_1V_BB+\phi_2V_PP-\frac{10}{\tau}\frac{n}{n+k}I-\alpha I
        	\label{eq:infected}
        \end{eqnarray}
        \begin{figure}[!ht]
            \centering
            \includegraphics[scale=0.75]{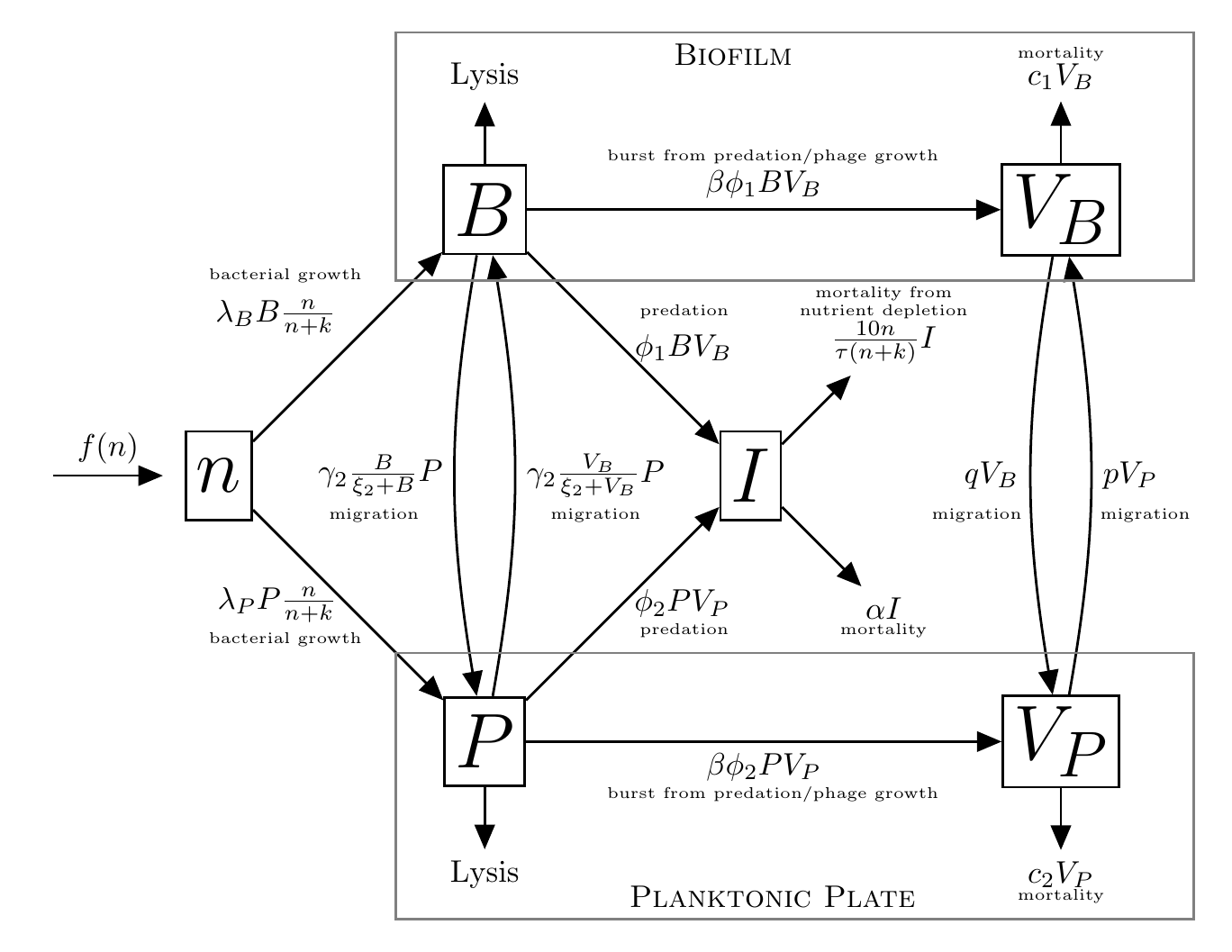}
            \caption{{\bf Flow diagram of the deterministic ODE model} This is a schematic representation 
            of the deterministic ODE model, showing the flows and connections with the parameters. These 
            parameters are also presented in \autoref{table1} with the actual descriptions of the parameter, 
            the values and sources. Because of lack of better word, we have used `migration' to capture the 
            conversion from biofilm to planktonic region, this does not imply any spatial component. }
            \label{fig2}
        \end{figure}

        Equation \eqref{eq:cellbiofilm} and \eqref{eq:cellplate} describe the bacterial growth in the 
        planktonic  and biofilm regions, the equations also describe predation of bacteria by the phages 
        and biofilm cell detachment which forms a direct coupling of the system. Equations 
        \eqref{eq:phagebiofilm} and \eqref{eq:phageplate} describes the phage growth in the biofilm and 
        planktonic region, while equation \eqref{eq:infected} keeps track of all infected bacterial cells in 
        both the biofilm and planktonic regions. Equation \eqref{eq:nutrients} describes the consumption of 
        nutrient by $B$ and $P$ respectively. The flow diagram of the model is presented in \autoref{fig2}. 
        The parameters on this model are presented in \autoref{table1}, one of the parameter of interest is 
        the burst size which determines the average number of phage release per bacterium; and could vary 
        from $10$ to $100$ for DNA transducing bacteriophages to about $20,000$ pfu for the RNA viruses. 

        \begin{table}[!ht]
            \centering
            \caption{
                {\bf Table of parameters}}
            \begin{tabular}{|l|l|l|l|l|}
                \hline
                \multicolumn{5}{|l|}{\bf Table of Parameters}\\ \thickhline
                \textbf{Symbol}& \textbf{Description} & \textbf{Value} & \textbf{Units} & \textbf{Source}  \\ 
                \hline 
                $n_{0}$ & Initial Nutrient concentration & $\text{variable}$ & $[gm^{-3}]$ & 
                \cite{refr19, refr15} \\ \hline
                $B_{0}$ & Initial biofilm bacteria & $\text{variable}$ & $[gm^{-3}]$ &  \cite{refr19, refr15} \\ 
                \hline
                $P_{0}$ & Initial planktonic bacteria & $\text{variable}$ & $[gm^{-3}]$ &  \cite{refr19}  \\ 
                \hline
                $V_{b0}$ & Initial biofilm phages & $\text{variable}$ & $[gm^{-3}]$ &   \text{Assumed} \\ 
                \hline
                $V_{p0}$ & Initial planktonic phages & $\text{variable}$ & $[gm^{-3}]$ & \cite{refr19} \\ 
                \hline
                $I_{0}$ & Initial Infected cells & $\text{variable}$ & $[gm^{-3}]$ &  \cite{refr19}  \\ 
                \hline
                $p$ & Phage detachement rate  & $0.1$ & $[d^{-1}]$ & \text{Assumed}  \\ 
                \hline
                $q$ & Phage reattach rate & $0.5$ & $[d^{-1}]$ & \text{Assumed} \\  
                \hline
                $\lambda_B$ & biofilm bacteria growth rate & $6.0$ & $[d^{-1}]$ &   \cite{refr19} \\ 
                \hline
                $\lambda_P$ & Planktonic bacteria growth rate  & $6.0$ & $[d^{-1}]$ &   \cite{refr19} \\ 
                \hline
                $\tau$ & Average latency time  & $0.5$ & $[h]$ &  \cite{refr19, refr20} \\ 
                \hline
                $k$ & Monod constant &$ 4.0$ & $[gm^{-3}]$ &    \cite{refr19} \\ 
                \hline
                $\alpha$ & Infection decay rate & $0.2$  & $[d^{-1}]$ & \text{Assumed}  \\ 
                \hline
                $\beta$ & Burst size  & $100$  & $[-]$ &    \cite{refr19} \\ 
                \hline
                $\gamma_1$ & Phage induced detach rate & $0.6$ & $[d^{-1}]$ &   \cite{refr15} \\ 
                \hline
                $\gamma_2$ & Natural detach rate & $0.3$   & $[d^{-1}]$ &   \cite{refr15} \\ \hline
                $\phi_1$ & Predation rate in biofilm & $10^{-8}$ &  $[m^3g^{-1}d^{-1}]$ & \text{Assumed}  \\ 
                \hline
                $\phi_2$ & Predation rate in planktonic & $10^{-6}$   & $[m^3g^{-1}d^{-1}]$ & \text{Assumed} \\ 
                \hline
                $\zeta_1$ & Monod saturation & $10^2$ & $[gm^{-3}]$ & \text{Assumed}  \\ 
                \hline
                $\zeta_2$ & Monod saturation & $10^4$ & $[gm^{-3}]$ & \text{Assumed} \\ \hline
                $c_1$ & Phage mortality rate in biofilm &   $2.1$ & $[d^{-1}]$ & \text{Assumed}    \\ 
                \hline 
                $c_2$ & Phage mortality rate in planktonic & $2.1$ & $[d^{-1}]$ & \text{Assumed}  \\ 
                \hline 
                $\mu_B$ & Bacteria mortality rate in Biofilm & $0.1$ & $[d^{-1}]$ & \text{Assumed} \\ 
                \hline
                $\mu_P$ & Bacteria mortality rate in planktonic & $0.1$ & $[d^{-1}]$ & \text{Assumed}  \\ 
                \hline
            \end{tabular}
            \label{table1}
        \end{table}
\newpage
    \subsubsection{\emph{Basic Reproduction Number}}
        The system \eqref{eq:nutrients}-\eqref{eq:infected} represents a nonlinear system of ODEs with the 
        interaction of phages and bacteria. 
        \begin{remark}
            The function $f(n)$, for the input/increase of nutrients, holds that 
            \begin{enumerate}
                \item $f(0)=0$ (thus making the origin an steady state) and $f'(0)>0$.
                \item For some fixed value $x>0$, the conditions $f(x)=0$ and $f'(x)<0$ are held. 
            \end{enumerate}
        \end{remark}
        
        Considering the above properties, we compute the Jacobian matrix which is given by:
        \begin{scriptsize}\beq
        J(X,\mathcal{P})= \left(\scalemath{0.65}{
        \barr{cccccc}
        	f'(n)-\frac{k(\lambda_BB+\lambda_PP)}{(n+k)^2} & -\frac{n\lambda_B}{n+k} & 
        	-\frac{n\lambda_P}{n+k} & 0 & 0 & 0 \\
        	\frac{kB\lambda_B}{(n+k)^2} & \dag &\frac{B\gamma_2}{\zeta_2+B} & 
        	-\phi_1B-\frac{\zeta_1\gamma_1B}{(\zeta_1+V_B)^2}&0&0\\
        	\frac{\lambda_PPk}{(n+k)^2}&\frac{\gamma_1V_B}{\zeta_1+V_B}-
        	\frac{\gamma_2\zeta_2P}{(\zeta_2+B)^2}&
        	\ddag&\frac{\gamma_1\zeta_1B}{(\zeta_1+V_B)^2}&-\phi_2P&0\\
        	0&\beta\phi_1V_B&0&\beta\phi_1B-c_1-q&p&0\\
        	0&0&\beta\phi_2V_P&q&\beta\phi_2P-p-c_2&0\\
        	-\frac{10kI}{\tau(n+k)^2}&\phi_1V_B&\phi_2V_P&\phi_1B&\phi_2P&
        	-\frac{10n}{\tau(n+k)}-\alpha
        \earr } \right) 
        \label{eq:JacobMatrix}
        \eeq\end{scriptsize}
        where 
        \bes\begin{split}
            \dag&=\frac{n\lambda_B}{n+k}-\phi_1V_B-\frac{\gamma_1V_B}{\zeta_1+V_B}
            +\frac{\gamma_2\zeta_2P}{(\zeta_2+B)^2}-\mu_B\\
            \ddag&=\frac{\lambda_Pn}{n+k}-\phi_2V_P-\frac{\gamma_2B}{\zeta_2+B}-\mu_P
        \end{split}\ees

        The equilibria points are determined by the zeros of the system 
        \eqref{eq:nutrients}-\eqref{eq:infected}. Considering the disease free equilibria which are 
        mainly determined in the absence of pathogen scenario, that is, when neither phages and infected 
        bacterial cells are present. From this context, we let $I=0$, $V_P=0$, and $V_B=0$; notice this 
        condition causes for \eqref{eq:phagebiofilm},\eqref{eq:phageplate}, and \eqref{eq:infected}
        to equal zero, thus reducing the equilibria problem to find the zeros of the system 
        \beq\begin{split}
	        \diff{d}{t}n&=f(n)-\left(\lambda_B B+\lambda_P P\right)\frac{n}{n+k}\\
	        \diff{d}{t}B&=B\left(\frac{\lambda_Bn}{n+k}+\frac{\gamma_2P}{\zeta_2+B}
	        -\mu_B\right)\\
	        \diff{d}{t}P&=P\left(\frac{\lambda_Pn}{n+k}-\frac{\gamma_2B}{\zeta_2+B}
	        -\mu_P\right)
        	\label{eq:d-fsystem}
        \end{split}\eeq
        Through a quick inspection we have that the equilibria points of \eqref{eq:d-fsystem} are 
        \begin{eqnarray}
	        DFE_1&=&\left(\scalemath{1}{n,0,0,0,0,0}\right)\label{eq:dfe1} ; \qquad\qquad 
	        \mbox{where $n$ is a zero of the function $f$} \nonumber \\
	        DFE_2&=&\left(\scalemath{0.85}{\frac{k\mu_B}{\lambda_B-\mu_B},
	        \frac{f\left(\frac{k\mu_B}{\lambda_B-\mu_B}\right)
	        \left(\frac{k\mu_B}{\lambda_B-\mu_B}+k\right)}{\lambda_B\frac{k\mu_B}{\lambda_B-\mu_B}},0,0,0,0}
	        \right)=\left(n_B^\ast,\frac{f(n_B^\ast)}{\mu_B},0,0,0,0\right)\label{eq:dfe2} \nonumber \\
	        DFE_3&=&\left(\scalemath{0.75}{\frac{k\mu_P}{\lambda_P-\mu_P},
	        0,\frac{f\left(\frac{k\mu_P}{\lambda_P-\mu_P}\right)\left(\frac{k\mu_P}{\lambda_P-\mu_P}+k\right)}
	        {\lambda_P\frac{k\mu_P}{\lambda_P-\mu_P}},0,0,0}\right)=\left(n_P^\ast,0,\frac{f(n_P^\ast)}
	        {\mu_P},0,0,0\right) \nonumber \label{eq:dfe3}
        \end{eqnarray}
        
        By evaluating \eqref{eq:JacobMatrix} at \eqref{eq:dfe1} we have that
        \beq
        	J_{0,1}=\left(\barr{cccccc}
            	f'(n)&-\frac{\lambda_Bn}{n+k}&-\frac{\lambda_Pn}{n+k}&0&0&0\\
            	0&\frac{\lambda_Bn}{n+k}-\mu_B&0&0&0&0\\
            	0&0&\frac{\lambda_Pn}{n+k}-\mu_P&0&0&0\\
            	0&0&0&-c_1-q&p&0\\
            	0&0&0&q&-c_2-p&0\\
            	0&0&0&0&0&-\frac{10n}{\tau(n+k)}-\alpha
        	\earr\right)\label{eq:dfe1jacobmat}
        \eeq
        which can be rewritten as the difference of two positive matrices $\mathcal{F}_1$ and 
        $\mathcal{V}_1$, thus we apply the next generation matrix. The basic reproductive number is 
        therefore given by the spectral radius of $\mathcal{F}_1\mathcal{V}_1^{-1}$, which is the eigenvalue 
        of largest magnitude, thus 
        \beq
            R_{0,1}=\max\left\{\frac{\lambda_Bn}{\mu_B(n+k)},\frac{\lambda_Pn}{\mu_P(n+k)},
            \pm\sqrt{\frac{pq}{(c_1+q)(c_2+p)}}\right\}\label{eq:dfe1R0}
        \eeq

        \begin{theorem}
            The Disease-Free equilibrium 1 (DFE1) is asymptotically stable if $n\neq0$ and $R_{0,1}<1$
        \end{theorem}

        \begin{proof}
        
            Suppose $n=0$, a zero of $f$, considering Remark 1, in this case the system naturally becomes 
            unstable since \eqref{eq:dfe1jacobmat} has positive eigenvalues, so it suffices that $n>0$ be a 
            zero of $f$.\\
            
            \noindent Remains to show $R_{0,1}<1$ for stability:  For this, the overall stability of 
            $DFE1=(n,0,0,0,0,0)$ is given by the roots of the characteristic polynomial of 
            \eqref{eq:dfe1jacobmat} which is
            \vspace{.1in}
            
            \noindent$det(\lambda I-J_{0,1})=$
            
            \vspace{.2in}
            	
            \noindent $det\left(\barr{cccccc}
            	\lambda-f'(n)&\frac{\lambda_Bn}{n+k}&\frac{\lambda_Pn}{n+k}&0&0&0\\
            	0&\lambda+\mu_B-\frac{\lambda_Bn}{n+k}&0&0&0&0\\
            	0&0&\lambda-\frac{\lambda_Pn}{n+k}+\mu_P&0&0&0\\
            	0&0&0&\lambda+c_1+q&-p&0\\
            	0&0&0&-q&\lambda+c_2+p&0\\
            	0&0&0&0&0&\lambda+\frac{10n}{\tau(n+k)}+\alpha
            	\earr\right) =$
            \vspace{.2in}
            	
            \noindent$(\lambda-f'(n))\left(\lambda+\mu_B-\frac{\lambda_Bn}{n+k}\right)\left(
            \lambda-\frac{\lambda_Pn}{n+k}+\mu_P\right)
            \left(\lambda+\frac{10n}{\tau(n+k)}+\alpha\right)\left(\left[\lambda+c_1+q\right]
            \left[\lambda+c_2+p\right]-pq\right)$
            
            \vspace{.2in}
            	
            \noindent whose roots are all negative if the following holds:
            \bes
                \mu_B>\frac{\lambda_Bn}{n+k}\,,\quad \frac{\lambda_Pn}{n+k}<\mu_P\,,\quad
                \left[\lambda+c_1+q\right]\left[\lambda+c_2+p\right]<pq
            \ees
            Notice that when $R_{0,1}<1$, the conditions above hold. Hence $DFE1$ is asymptotically stable.
        \end{proof}

        \vspace{.1in}

        \noindent Similarly, for the second disease free equilibrium (DFE2) we have that 

        \vspace{.1in}

        \noindent $J_{0,2}= \left(\barr{cccccc}
	        f'(n_B^\ast)- k^\star & - q^\star\lambda_B  & -q^\star\lambda_P & 0 & 0 & 0\\
        	k^\star & q^\star\lambda_B -\mu_B & 
        	s^\star &
        	-t^\star\phi_1  -\frac{\gamma_1t^\star }{\zeta_1}&0&0\\
        	0&0& q^\star\lambda_P  -s^\star \mu_P &\frac{\gamma_1t^\star}{\zeta_1}&0&0\\
        	0&0&0&\beta t^\star \phi_1-c_1-q&p&0\\
        	0&0&0&q&-p-c_2&0\\
        	0&0&0&t^\star \phi_1&0&
        	-\frac{10q^\star}{\tau}-\alpha
        \earr\right) $

        \vspace{.1in}

        $k^\star = \frac{k\lambda_Bf(n_B^\ast)}{\mu_B(n_B^\ast+k)^2}$, 
        $ q^\star = \frac{n_B^\ast}{n_B^\ast+k}$, 
        $s^\star = \frac{f(n_B^\ast)\gamma_2}{\mu_B\zeta_2+f(n_B^\ast)} $, 
        $t^\star = \frac{f(n_B^\ast)}{\mu_B} $

        \vspace{.1in}

        \noindent which can be rewritten as the difference of two positive matrices $\mathcal{F}_1$ and 
        $\mathcal{V}_1$, thus we apply the next generation matrix. The basic reproductive number is 
        therefore given by the spectral radius of $\mathcal{F}_1\mathcal{V}_1^{-1}$ which is the eigenvalue 
        of largest magnitude, thus 

        \begin{tiny}
            \beq
                R_{0,2}=\max\left\{\frac{k\lambda_B f(n_B^*)}{ \lvert f^\prime(n_B^*) \rvert 
                \mu_B(n_B^* + k)^2 + k\lambda_Bf(n_B^*)},\frac{(\mu_B\xi_2 + 
                f(n_B^*))\lambda_P\eta_B^*}{(n_B^*+k)\{\gamma_2f(n_B^*)+\mu_P(\mu_B\xi_2 + f(n_B^*))\}},
                \frac{\phi_1f(n_B^*)}{\mu_B(c_1+q)}
                \right\}\label{eq:dfe2R0}
            \eeq
        \end{tiny}

        \begin{theorem}
            The Disease-Free equilibrium 2 (DFE2) is asymptotically stable if $n_B^\ast\neq0$, 
            $f^\prime(n_B^\ast) \gg 0$ and $R_{0,2}<1$
        \end{theorem}

        \begin{proof}
            Suppose $n=0$, a zero of $f$, considering Remark 1, in this case the system naturally 
            becomes unstable since the DFE2 Jacobian matrix has positive eigenvalues, so it suffices 
            that $n>0$ be a zero of $f$.  \\
            
            \noindent Remains to show $R_{0,2}<1$ for stability:  For this, the overall stability of 
            $DFE2=\left(n_B^\ast,\frac{f(n_B^\ast)}{\mu_B},0,0,0,0\right)$ is given by the roots of 
            the characteristic polynomial of DFE2 Jacobian matrix which is
            \vspace{.1in}
        
            \noindent $det(\lambda I-J_{0,2})=$
        
            \vspace{.2in}
        
            \begin{tiny}
                \noindent $det \left(\barr{cccccc}
        	    \lambda-f'(n_B^\ast)+ k^\star &  q^\star\lambda_B  & q^\star\lambda_P & 0 & 0 & 0\\
            	-k^\star & \lambda-q^\star\lambda_B +\mu_B & 
            	-s^\star &
            	t^\star\phi_1  + \frac{\gamma_1t^\star }{\zeta_1}&0&0\\
            	0&0& \lambda-q^\star\lambda_P + s^\star \mu_P &\frac{\gamma_1t^\star}{\zeta_1}&0&0\\
            	0&0&0& \lambda- \beta t^\star \phi_1+c_1+q&p&0\\
            	0&0&0&q& \lambda+p+c_2&0\\
            	0&0&0& -t^\star \phi_1&0&
            	\lambda+\frac{10q^\star}{\tau}+\alpha\earr\right) $
            \end{tiny}
        
            \vspace{.2in}
        
            \noindent where, 
            
            \vspace{.1in}
            
            $k^\star = \frac{k\lambda_Bf(n_B^\ast)}{\mu_B(n_B^\ast+k)^2}$, 
            $ q^\star = \frac{n_B^\ast}{n_B^\ast+k}$, 
            $s^\star = \frac{f(n_B^\ast)\gamma_2}{\mu_B\zeta_2+f(n_B^\ast)} $, 
            $t^\star = \frac{f(n_B^\ast)}{\mu_B} $

        
            \noindent whose roots are all negative if the following holds:  
        
            \begin{itemize}
                \item If $f'(n_B^\ast)+\frac{n_B^\ast\lambda_B}{n_B^\ast+k}<
                \frac{k\lambda_Bf(n_B^\ast)}{\mu_B(n_B^\ast+k)^2}+\mu_B$, 
                \begin{itemize}
                    \item since phages are gone, and bacteria is growing, we expect the mortality rate 
                    of bacteria to be negligible when compared with the growth rate.
                \end{itemize}
                \item $f'(n_B^\ast)+\frac{n_B^\ast\lambda_B}{n_B^\ast+k}<
                \frac{k\lambda_Bf(n_B^\ast)}{\mu_B(n_B^\ast+k)^2}+\mu_B$
                \begin{itemize}
                    \item this means that the bacteria growth in the biofilm is greater than the 
                    bacteria death, this sounds right since DFE of phages in the biofilm implies 
                    loosing more phages and having more bacteria
                \end{itemize}
                \item $(c_2+p)c_1+qc_2 >\frac{\beta\phi_1f(n_B^\ast)(c_2+p)}{\mu_P} $
                \begin{itemize}
                    \item this means that the phage loss in the biofilm is greater than phage growth
                \end{itemize}
            \end{itemize}
        
            Notice that when $R_{0,2}<1$, the conditions above hold. Hence $DFE2$ is asymptotically 
            stable. This concludes the proof.
        \end{proof}
        \noindent Finally, for the third disease free equilibrium (DFE3) we have that 

        \vspace{.1in}

        \begin{tiny}
            $J_{0,3}=$
            \bes
            \left(\barr{cccccc}
            	f'(n_P^\ast)-\frac{k\lambda_Pf(n_P^\ast)}{\mu_P(n_P^\ast+k)^2} & 
            	-\frac{n_P^\ast\lambda_B}{n_P^\ast+k} & 
            	-\frac{n_P^\ast\lambda_P}{n_P^\ast+k} & 0 & 0 & 0\\
            	0 & \frac{n_P^\ast\lambda_B}{n_P^\ast+k}+\frac{\gamma_2f(n_P^\ast)}{\mu_P\zeta_2}-\mu_B 
            	& 0 & 0 & 0 & 0 \\
            	\frac{\lambda_Pkf(n_P^\ast)}{\mu_P(n_P^\ast+k)^2}
            	&-\frac{\gamma_2f(n_P^\ast)}{\mu_P\zeta_2}&
            	\frac{\lambda_Pn_P^\ast}{n_P^\ast+k}-\mu_P
            	&0&-\phi_2\frac{f(n_P^\ast)}{\mu_P}&0\\
            	0&0&0&-c_1-q&p&0\\
            	0&0&0&q&\beta\phi_2\frac{f(n_P^\ast)}{\mu_P}-p-c_2&0\\
            	0&0&0&0&\phi_2\frac{f(n_P^\ast)}{\mu_P}&
            	-\frac{10n_P^\ast}{\tau(n_P^\ast+k)}-\alpha
        	\earr\right)
            \ees
        \end{tiny}

        \vspace{.1in}

        \noindent which can be rewritten as the difference of two positive matrices $\mathcal{F}_1$ and 
        $\mathcal{V}_1$, thus we apply the next generation matrix. The basic reproductive number is 
        therefore given by the spectral radius of $\mathcal{F}_1\mathcal{V}_1^{-1}$ which is the 
        eigenvalue of largest magnitude, thus 

        \vspace{.01in}
        \begin{tiny}
            \beq\begin{split}
                R_{0,3}=\max\left\{
                \frac{n_P^\ast\lambda_B}{\mu_B(n_P^\ast+k)}
                +\frac{\gamma_2f(n_P^\ast)}{\mu_P\mu_B\zeta_2},\right.&\left.
                \frac{\lambda_Pn_P^\ast}{\mu_P(n_P^\ast+k)}\left[
                1-\frac{k\lambda_Pf(n_P^\ast)}
                {\lvert f'(n_P^\ast) \rvert \mu_P(n_P^\ast+k)^2+k\lambda_Pf(n_P^\ast)}\right],\right.\\
                &\qquad\left.\frac{\beta\phi_2f(n_P^\ast)}{2\mu_P(c_2+p)}\pm
                \sqrt{\left(\frac{\beta\phi_2f(n_P^\ast)}{2\mu_P(c_2+p)}\right)^2+\frac{pq}{(c_1+q)(c_2+p)}}
                \right\}
                \label{eq:dfe3R0}
            \end{split}\eeq
        \end{tiny}
        \begin{theorem}
            The Disease-Free equilibrium 3 (DFE3) is asymptotically stable if $n\neq0$, 
            $f^\prime(n) \gg 0$ and $R_{0,3}<1$
        \end{theorem}

        \begin{proof}
            Suppose $n=0$, a zero of $f$, considering Remark 1, in this case the system naturally becomes 
            unstable since the DFE3 Jacobian matrix has positive eigenvalues, so it suffices that $n>0$ 
            be a zero of $f$.  \\
        
            \noindent Remains to show $R_{0,3}<1$ for stability:  For this, the overall stability of 
            $DFE3=\left(n_P^\ast,0,\frac{f(n_P^\ast)}{\mu_P},0,0,0\right)$ is given by the roots of 
            the characteristic polynomial of DFE3 Jacobian matrix which is
            \vspace{.1in}
            
            \begin{tiny}
                \bes\begin{split}
                    det(\lambda I-J_{0,3})=&\left\lvert\barr{ccc}
            	    \lambda-f'(n_P^\ast)+\frac{k\lambda_Pf(n_P^\ast)}{\mu_P(n_P^\ast+k)^2} & 
            	    \frac{n_P^\ast\lambda_B}{n_P^\ast+k} & \frac{n_P^\ast\lambda_P}{n_P^\ast+k} \\
            	    0 & \lambda-\frac{n_P^\ast\lambda_B}{n_P^\ast+k}
            	    -\frac{\gamma_2f(n_P^\ast)}{\mu_P\zeta_2}+\mu_B & 0\\
            	    -\frac{\lambda_Pkf(n_P^\ast)}{\mu_P(n_P^\ast+k)^2}&
            	    \frac{\gamma_2f(n_P^\ast)}{\mu_P\zeta_2}&
            	    \lambda-\frac{\lambda_Pn_P^\ast}{n_P^\ast+k}+\mu_P
            	    \earr\right\rvert
            	    \times\\&\qquad\qquad
            	    \left\lvert\barr{ccc}
            	    \lambda+c_1+q&-p&0\\
            	    -q&\lambda-\beta\phi_2\frac{f(n_P^\ast)}{\mu_P}+p+c_2&0\\
            	    0&-\phi_2\frac{f(n_P^\ast)}{\mu_P}&
            	    \lambda+\frac{10n_P^\ast}{\tau(n_P^\ast+k)}+\alpha
            	    \earr\right\rvert
            	    \label{eq:dfe1charpoly}
                \end{split}\ees
            \end{tiny}
            whose roots are all negative if the following holds:
            \begin{itemize}
                \item $\displaystyle 
                \frac{n_P^\ast\lambda_B}{n_P^\ast+k}+\frac{\gamma_2f(n_P^\ast)}{\mu_P\zeta_2}<\mu_B$
                \item $f'(n_P^\ast)+\frac{\lambda_Pn_P^\ast}{n_P^\ast+k}<\mu_P+\frac{k\lambda_Pf(n_P^\ast)}
                {\mu_P(n_P^\ast+k)^2}$ and $f'(n_P^\ast)\frac{\lambda_Pn_P^\ast}{n_P^\ast+k}
                +\frac{k\lambda_Pf(n_P^\ast)}{(n_P^\ast+k)^2}>\mu_Pf'(n_P^\ast)$
                \item $\beta\phi_2\frac{f(n_P^\ast)}{\mu_P}<p+c_2+q+c_1$ and 
                $\beta\phi_2\frac{f(n_P^\ast)}{\mu_P}<\frac{pc_1}{c_1+q}+c_2$
            \end{itemize}
            when $R_{0,3}<1$ these conditions hold. Hence $DFE3$ is asymptotically stable, concluding this 
            proof.
            \vspace{.2in}
        \end{proof}

    \subsection{\emph{ODE Model calibration}}
        All the computations were done using Matlab, the parameter values of the model equations are 
        listed in \autoref{table1} which shows the parameter values, units and sources. Due to the 
        novelty of this study, some of the parameter values are assumed based on similar studies such 
        as \cite{refr15}. Parameter sensitivity were performed in the subsequent section to determine 
        which of the parameter values have strong impact on the model generally. We let the program 
        to run for at least $15$ units of time in order to see the model behaviour in its completeness. 

        \subsubsection*{Phage burst size controls Biofilm growth} 
            The first simulation experiments investigates the effect of phage burst size on the 
            biofilm growth. Here, we have considered a situation of steady supply of nutrients to the 
            biofilm, and we have varied the burst size as $10$, $50$, $150$, and $250$. These results 
            are presented in \autoref{fig3}. These simulations reveal that the bacteria cells reduces 
            so quickly as the burst size increases and the biofilm growth balances over time; the 
            biofilm phage population increases significantly as the burst size increases, the phage 
            growth balances after a while. A similar outcome is also seen in the planktonic region. 
            Several studies show that Planktonic cells grow more rapidly than bacteria cells within 
            a biofilm, therefore the phage burst size in the biofilm could be several-fold 
            smaller and the infection cycle takes even longer 
            \cite{refr21,refr22,refr23,refr24,refr25,refr26,refr27, refr28}. Remarkably there is a 
            standard procedure for determination of phage burst size, which is defined as the number of 
            phage progeny produced per infected bacterial cell \cite{refr29, refr30, refr31, refr32}. 
            Phage burst size differ from phage to phage depending on the lysis time.

            \begin{figure}[!ht]
                \includegraphics[scale=0.3]{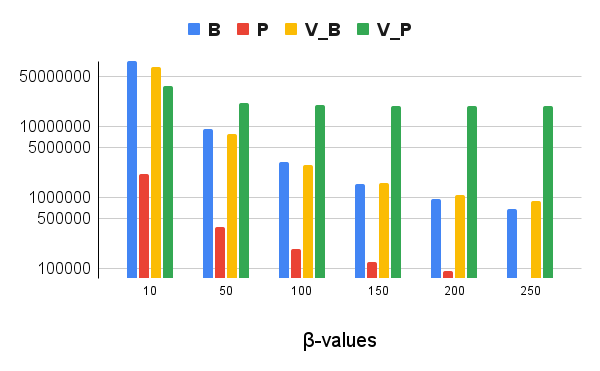} \includegraphics[scale=0.3]{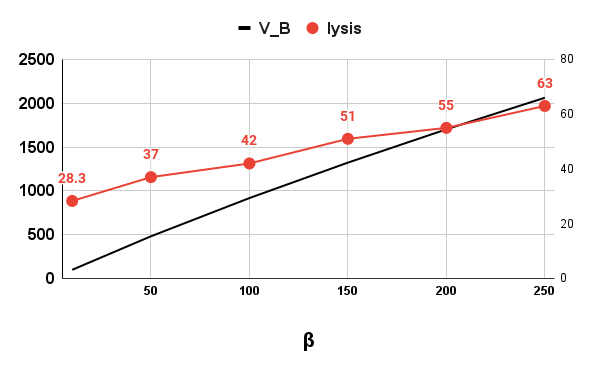}
                \caption{{\bf Variability of Phage Burst Size:} Left - the maximum value of the biofilm 
                cells $B$, planktonic cells $P$, Biofilm phage $V_B$ and the planktonic phage $V_P$, this 
                is plotted for different values of the burst size $\beta$. Right - The minimum possible 
                value of the biofilm phage for different burst size (black solid lines) is compared with 
                the data from \cite{refr33} for different lysis time for comparable burst sizes (red) .}
                \label{fig3}
            \end{figure}
    \subsection*{\emph{Stochastic Model - CTMC}}
        If the bacteria (or phage) population is sufficiently small, an ordinary differential equation 
        model is not appropriate, hence we utilize a continuous-time Markov chain (CTMC) model, which is 
        continuous in time and discrete in the state space in order to study the variability at the 
        initiation of bacteria clearance during phage treatment therapy, peak level of phage infection 
        (in the phage-bacteria interaction, phages are seen as the pathogen and the bacteria are the 
        susceptible). To make it simple, we use the same notation for the state variables as in the 
        ordinary differential equation. The state variables are discrete random variables, 
        $n, B, I, P \in \{0,1,2,3\}$ and $t\in [0, \infty]$
    
        To formulate the CTMC, it is necessary to define the infinitesimal transition probabilities that 
        corresponds to each event in the state variables, this is outlined in \autoref{ctmctable} which 
        consists of $17$ distinct events.
    
        \begin{table}[!ht]
            \centering
            \caption{Table of transitions and corresponding probabilities in stochastic model}
            \begin{tabular}{|p{0.6cm}|p{2.5cm}|c|p{3cm}|p{3cm}|} 
                \hline
                \multicolumn{5}{|l|}{\bf Table of events}\\ \thickhline
                \hline \textbf{Event} & \textbf{Event Description} & \textbf{Transitions} & 
                \textbf{Change} \scriptsize{$\scalemath{0.75}{(\Delta n, \Delta B,\Delta P, \Delta Vb, 
                \Delta Vp, \Delta I)}$} & \textbf{Probability} \\ \hline
                1 & Availability of nutrient & $n \rightarrow n + 1$ & $(1, 0, 0,0,0,0)$ & 
                $f(n)\Delta t +o(\Delta t)$\\ \hline
                2 & Nutrient consumption and bacterial growth & $n \rightarrow n - 1$ & 
                $(-1, 1, 0,0,0,0)$ & $\left(\lambda_B B\right)\frac{n}{n+k}\Delta t +o(\Delta t)$\\  
                & & $B \rightarrow B + 1$  &  & \\ \hline
                3 & Nutrient consumption and bacterial growth & $n \rightarrow n - 1$ & 
                $(-1, 0, 1,0,0,0)$ & $\left(\lambda_P P\right)\frac{n}{n+k}\Delta t +o(\Delta t)$\\  
                & & $P \rightarrow P + 1$  &  & \\  \hline
                4 & Bacteria migration & $B \rightarrow B - 1$ & $(0, -1, +1,0,0,0)$ & 
                $(\gamma_1\frac{V_B}{\zeta_1+V_B}B ) \Delta t +o(\Delta t)$ \\ 
                & & $P \rightarrow P + 1$ & & \\ \hline
                5 & Bacteria migration & $P \rightarrow P - 1$ & $(0, 1, -1,0,0,0)$ & 
                $( \gamma_2\frac{B}{\zeta_2+B}P) \Delta t +o(\Delta t)$ \\ 
                & & $B \rightarrow B + 1$ & & \\ \hline
                6 & Biofilm Bacteria infection by phage & $B \rightarrow B - 1$ & $(0,-1, 0, 0,0,1)$ & 
                $(\phi_1BV_B ) \Delta t +o(\Delta t)$ \\
                & & $I \rightarrow I + 1$ & & \\ \hline
                7 & Planktonic Bacteria infection by phage & $P \rightarrow P - 1$ & $(0,0, -1, 0,0,1)$ & 
                $( \phi_2PV_P) \Delta t +o(\Delta t)$ \\
                & & $I \rightarrow I + 1$ & & \\ \hline
                8 & Biofilm-Phage Migration & $V_B \rightarrow V_B - 1$ & $(0,0,0,-1, 1, 0)$ & 
                $(qV_B) \Delta t +o(\Delta t)$ \\
                & & $V_P \rightarrow V_P + 1$ & &  \\ \hline
                9 & Biofilm-Phage Migration & $V_P \rightarrow V_P - 1$ & $(0,0,0,1, -1, 0)$ & 
                $(pV_P) \Delta t +o(\Delta t)$ \\
                & & $V_P \rightarrow V_P + 1$ & &  \\ \hline
                10 & Biofilm-Phages gain from infected cells & $V_B \rightarrow V_B + 1$ & 
                $(0, 0, 0,1,0,0)$ & $(\beta\phi_1V_BB )\Delta t + o(\Delta t)$ \\\hline
                11 & Planktonic-Phages gain from infected cells & $V_P \rightarrow V_P + 1$ & 
                $(0, 0, 0,0,1,0)$ & $( \beta\phi_2V_PP)\Delta t + o(\Delta t)$ \\\hline
                12 & Death of Biofilm bacteria & $B \rightarrow B - 1$ & $(0, -1, 0,0,0,0)$ & 
                $(\mu_BB) \Delta t +o(\Delta t)$ \\ \hline
                13 & Death of Planktonic bacteria & $P \rightarrow P - 1$ & $(0, 0, -1,0,0,0)$ & 
                $(\mu_PP) \Delta t +o(\Delta t)$\\ \hline
                14 & Biofilm-Phages death & $V_P \rightarrow V_P - 1$ & $(0, 0, 0,-1,0,0)$ & 
                $(c_1V_BB )\Delta t + o(\Delta t)$ \\\hline
                15 & Planktonic-Phages death & $V_B \rightarrow V_B - 1$ & $(0, 0, 0,-1,0,0)$ & 
                $( c_2V_PP)\Delta t + o(\Delta t)$ \\\hline
                16 & Decay of infected cells & $I \rightarrow I - 1$ & $(0, 0, 0,0,0,-1)$ & 
                $(\frac{10n}{\tau(n+k)}) \Delta t + o(\Delta t)$ \\\hline
                17 & Death of infected cells & $I \rightarrow I - 1$ & $(0, 0, 0,0,0,-1)$ & 
                $(\alpha I) \Delta t + o(\Delta t)$ \\\hline
            \end{tabular}
            \label{ctmctable}
        \end{table}
    
        \subsubsection{CTMC Analysis}
            For the continuous time markov chains, we numerically simulate the sample paths in order to 
            determine the peak number of infected bacteria and peak phage-bacteria infection. For the 
            sample paths, we simply compare our results with that of the ordinary differential equations.
        
        \subsubsection*{Sample paths}
            An example of the sample paths that result from the Continuous Time Markov Chain model is 
            shown in \autoref{samplepaths} where the sample paths are captured by the red, blue, green 
            and yellow, whereas the ODE model is captured by the black line. We observe that these sample 
            paths generally aligned with  the population average response that is captured by the ODE model.  
            The sample paths of the CTMC model show the potential variability in timing of the peak level of 
            infection and the peak number of infected bacteria. Due to limitation in computational memory, 
            we reduced the initial values of the dependent variables as shown in \autoref{ctmctable}; in 
            other words, we have used the following initial conditions 
            $n_0=10^5 gm^{-3}, B_0=10^4gm^{-3}, P_0=0, V_{B0}=10^3gm^{-3}, V_{P0}=0,I_0=0$
        

            \begin{figure}[!ht]
                \includegraphics[scale=0.6]{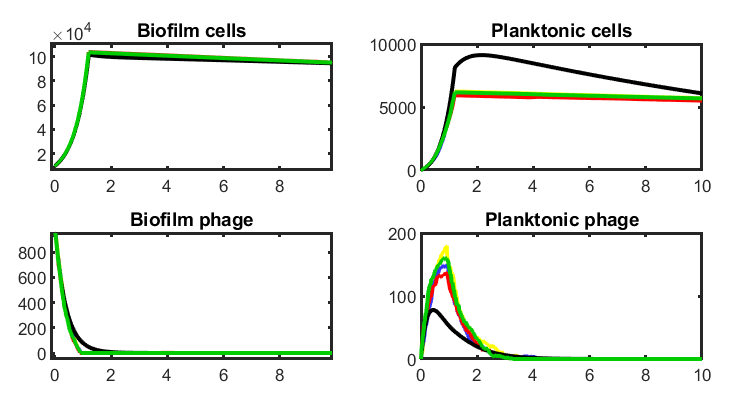}
                \caption{{\bf Sample Paths} This is the sample paths to the CMTC model, due to 
                insufficient memory space, we used a reduced initial conditions for all the dependent 
                variables, that is $1e5, 1e2, 0, 1e3 , 0 , 0$ corresponding to 
                $n_0, B_0 ,P_0, Vb_0 , Vp_0 , I_0$ respectively}
                \label{samplepaths}
            \end{figure}
        
        \subsubsection*{Time to Peak Infection and Peak Number of Infected Biofilm Bacteria}
            We asked whether bacteria infection will reach peak infection in a shorter time, to 
            investigate this, we calculated the mean ($\pm$ SD) of time to peak infection for the 
            bacteria in the biofilm. This is presented in \autoref{fig5} showing that we can attend 
            peak infection with just one phage in the system within a limited time. By introducing 
            few bacteriophages, we observed a large amount of infected bacterial cells resulting from 
            the interaction, this shows that there were a large replication of the viruses. 
            Interestingly, this happened within a short period of time. Even though we do not know 
            how long it might take phage therapy to work, experimental data have shown that treatment 
            of bacteria infection could be achieved in a period as short as $10$ days and up to $8$ 
            weeks \cite{refr38}, this is in agreement with our finding. 
            \begin{figure}[!ht]
                \includegraphics[scale=0.4]{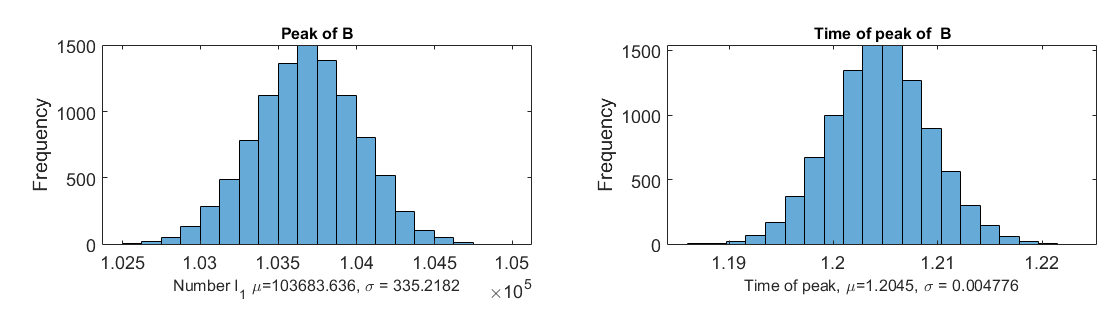}
                \caption{{\bf Peak Infection} This shows the peak infection of bacteria (left) and 
                the corresponding time to reach the peak infection (right).}
                \label{fig5}
            \end{figure}
        
    \section{Parameter Sensitivity Analysis}
        We perform a sensitivity analysis on the parameters ranges given in \autoref{tab:params} for 
        the ODE models using a uniform distribution for the values. Latin hypercube sampling (LHS), 
        first developed by McKay et al. \cite{refr31, refr34}, with the statistical sensitivity measure 
        partial rank correlation coefficient (PRCC), performs a sensitivity analysis that explores a 
        defined parameter space of the model. The parameter space considered is defined by the parameter 
        intervals depicted in \autoref{tab:params}. Rather than simply exploring one parameter at a 
        time with other parameters held fixed at baseline values, the LHS/PRCC sensitivity analysis 
        method globally explores multidimensional parameter space. LHS is a stratified Monte Carlo 
        sampling without replacement technique that allows an unbiased estimate of the average model 
        output with limited samples. The PRCC sensitivity analysis technique works well for parameters 
        that have a nonlinear and monotonic relationship with the output measure. The PRCC presented in 
        \autoref{fig6} shows how the output measure is influenced by changes in a specific parameter value 
        when the linear effects of other parameter values are removed. The PRCC values were calculated as 
        Spearman (rank) partial correlations using the partialcorr function in MATLAB 2020. Their 
        significance, uncorrelated p-values, were also determined. The PRCC values vary between $-1$ and $1$, 
        where negative values indicate that the parameter is inversely proportional to the outcome measure. 
        Following Marino et al. \cite{refr35}, we performed a z-test on transformed PRCC values to rank 
        significant model parameters in terms of relative sensitivity. According to the z-test, parameters 
        with larger magnitude values had a stronger effect on the output measures.

        We start by verifying the monotonicity of the output measures. Monotonicity was observed for all 
        parameters, hence we use PRCC. PRCC analysis of these ranges produces similar results. For the 
        biofilm-phage model, we calculate the PRCC for the following output measures: infected bacteria 
        cells, bacteriophages in the planktonic phase and 10\% population of infected bacteria cells. 
        \begin{table}[htp]
            \begin{center}
                \begin{tabular}{|l|l| c|c |c|c|}
             	    \hline
             	    & Parameter & \multicolumn{2}{c|}{Values} \\ 
             	    \cline{3-4}
            	    {}&{} & Baseline & Range \\ \hline 
                	$\lambda_B$ & biofilm bacteria growth rate & 0.06   &   \\ 
                	$\lambda_P$ & planktonic bacteria growth rate & 0.6931$^{**}$  & (0.4, 0.8)    \\
                	$k$ & Monod constant & 6.3   &(2-8)   \\
                	$\gamma_1$ &  phage induced detachment rate &  0.4  &(0.1, 2)    \\
                	$\gamma_2$ & natural detachment rate & 0.08  &(0.02, 0.14)   \\
                	$\phi_1$ & biofilm predation rate & 0.075  &(0.05, 0.1)   \\ 
                	$\phi_2$ & planktonic predation rate & 0.075  &(0.05, 0.1)   \\ 
                	$\zeta_1$ & Monod saturation & 0.075  &(0.05, 0.1)   \\ 
                	$\zeta_2$ & Monod saturation & 0.075  &(0.05, 0.1)   \\ 
                	$c_1$ & biofilm phage lysis & 0.075  &(0.05, 0.1)   \\ 
                	$c_2$ & planktonic phage lysis & 0.075  &(0.05, 0.1)   \\ 	
                	$p$ & phage detachment rate & 6.3   &(2-8)   \\
                	$q$ & phage reattachment rate & 6.3   &(2-8)   \\
                	$\tau$ & Average latency time & 6.3   &(2-8)   \\
                	$b$ & Burst size & 6.3   &(2-8)   \\
                	$\alpha$ & Infection decay rate & 6.3   &(2-8)   \\
                	\hline
                \end{tabular}
                \caption{\label{tab:params} Baseline parameter values are used in all simulations. 
                The range of values presented in this table is used  for the parameter sensitivity analysis}
            \end{center}
        \end{table}


        \begin{figure}[!ht]
            \includegraphics[scale=0.4]{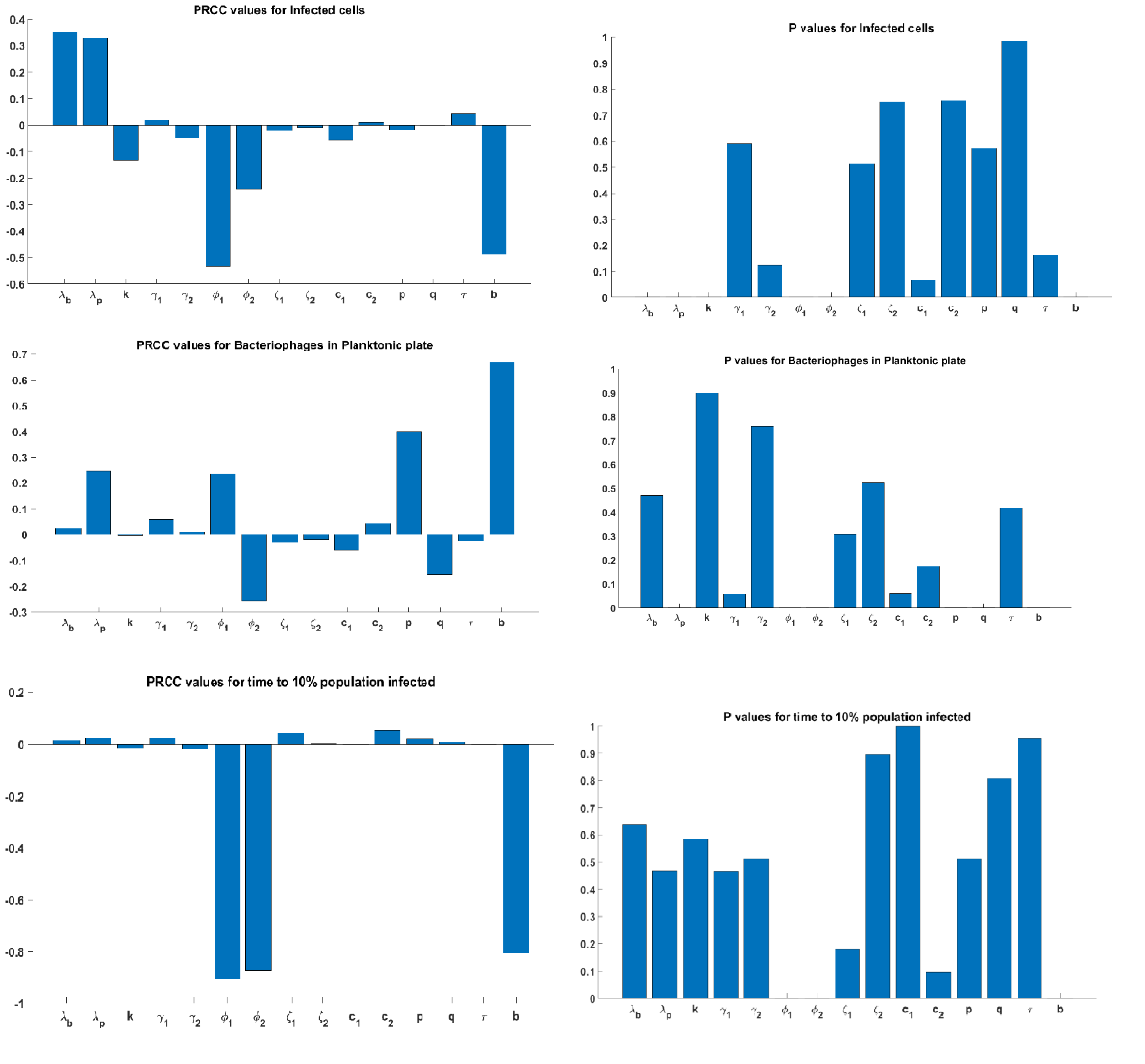}
            \caption{{\bf PRCC} This is the partial rank correlation coefficient for the sensitivity 
            analysis of the parameters}
            \label{fig6}
        \end{figure}



    \section{Discussion}
        We have developed a bacteria-phage interaction model within a biofilm in a cystic fibrosis patient. 
        The model considers bacteria in a biofilm and planktonic phase. The model assumes that the 
        interactions are region-specific, which means that the phages in the planktonic phase can only 
        interact with the planktonic bacteria while the phages in the biofilm can only interact with 
        biofilm bacteria cells. The model speculates that the burst size could control the biofilm growth. 
        In an effort to understand the state of the disease over time, we developed a stochastic model 
        with which we could investigate the probabilities of reaching peak infection within a short time.

        The model in this study can be easily adopted to investigate the effect of factors such as 
        temperature and pH value on middle ear infection. Experimental study in \cite{osgood, osgoodref} 
        revealed that the pH of middle ear fluid collected from acute otitis media  of children could affect 
        biofilm formation, and biofilm formation is limited or completely absent under aerobic conditions 
        as likely to happen, therefore the current model in this study can be adopted with the inclusion 
        of these specific factors to understand the interaction of phages and bacteria in middle ear 
        infection.

        Our assumptions and findings are consistent with the dynamics associated with biofilms. 
        For instance, one of our main assumptions confirms that the phage interaction rate in the biofilm 
        is different from the planktonic since the bacteria in the biofilm are dense and encased with 
        extracellular polymeric substances, this is consistent with several in vitro settings 
        \cite{refr6, refr13, refr17, refr18, refr19, refr36, refr37}.

        In connecting models to experiment, our model is not able to explain the biofilm occupancy which 
        will require the spatial components incorporated into the model, the spatial structures can 
        definitely be therapeutically relevant. For example,  understanding the biofilm matrix and EPS 
        will help to understand the actual interaction rates within the biofilm; the bacteria occupancy 
        in the biofilm will help to determine if the phage interactions is at the biofilm surface, 
        mimicking a lollipop-like degradation, or from within the biofilm, thus forming cavities. 
        Another extension of this model are to: (a) investigate a combination therapy that will 
        involve antibiotics and immune response, (b) investigate the factors that influence biofilm 
        formation and how they can be manipulated to prevent and eliminate biofilm-associated diseases 
        in other areas of the body.

\end{document}